\documentclass{llncs}

\usepackage[utf8]{inputenc}
\usepackage[T1]{fontenc}
\usepackage{amsmath}
\usepackage{amssymb}
\usepackage{mathtools}
\usepackage{parskip}
\usepackage{tikz}
\usepackage{subcaption}
\usepackage{graphicx}
\usepackage{float}
\usepackage{doi}
\usepackage{appendix}

\graphicspath{ {images} }
\DeclareGraphicsExtensions{.pdf,.png}

\usetikzlibrary{arrows,shapes,fit}
\pgfdeclarelayer{background}
\pgfsetlayers{background,main}

\DeclarePairedDelimiter{\set}{\{}{\}}
\DeclarePairedDelimiter{\tuple}{(}{)}
\DeclarePairedDelimiter{\abs}{\lvert}{\rvert}

\let\oldleq\leq
\renewcommand{\leq}[1][]{\oldleq_{#1}}
\renewcommand{\implies}{\rightarrow}

\newcommand{\bicond}{\leftrightarrow}
\newcommand{\poset}[1]{\mathcal{#1}}
\newcommand{\uni}[1][]{\Omega_{#1}}
\newcommand{\lang}[1]{L(#1)}
\newcommand{\lin}[1]{\texttt{#1}}
\newcommand{\swap}[1][]{\leftrightarrow_{#1}}
\newcommand{\sgraph}[1]{G(#1)}
\newcommand{\lext}{\sqsubseteq}
\newcommand{\incomp}[1][]{\parallel_{#1}}
\newcommand{\covered}[1][]{\prec_{#1}}
\newcommand{\complmt}[1]{\overline{#1}}
\newcommand{\satvar}[2]{\mathtt{X}_{#1}^{#2}}
\newcommand{\bigo}[1]{\mathcal{O}(#1)}

\newcommand{\swapfn}[2]{#1[#2]}
\newcommand{\inv}[2]{inv(#1,#2)}

\begin{document}

\title{
    SAT-Solving the Poset Cover Problem
    \thanks{This work was conducted in 2017--2018 while the first author
    was an undergraduate student at National Taiwan Normal University
    and an adjunct research assistant at Academia Sinica.}
}

\author{Chih-Cheng Rex Yuan \and
        Bow-Yaw Wang}

\institute{
    Institute of Information Science, Academia Sinica, Taipei, Taiwan \\
    \email{hello@rexyuan.com, bywang@iis.sinica.edu.tw}
}

\maketitle

\begin{abstract}
The poset cover problem seeks a minimum set of partial orders whose linear extensions cover a given set of linear orders. Recognizing its NP-completeness, we devised a non-trivial reduction to the Boolean satisfiability problem using a technique we call swap graphs, which avoids the complexity explosion of the naive method. By leveraging modern SAT solvers, we efficiently solve instances with reasonable universe sizes. Experimental results using the Z3 theorem prover on randomly generated inputs demonstrate the effectiveness of our method.

\end{abstract}

\section{Introduction}

Orders have been in our lives as long as humans can order things. We organize our tasks, prioritize our goals, and rank our preferences. Orders are everywhere in our world.

Partial orders, a generalization of orders, are a mathematical construct that captures the essence of non-linear ordering. For example, we can order both the fields of physics and chemistry under the general category of natural science, but we cannot say one is under the other.

We can find partial orders in various places in the wild. Scheduling tasks and analyzing biological sequences are but a few examples of partial orders that find themselves in\cite{heath2013poset,fernandez2013mining}.

Mathematicians have formalized the concept of partial orders as a set of elements with a binary relation that is reflexive, antisymmetric, and transitive. This structure is called a partially ordered set, or \emph{poset} for short.

By a foundational theorem called Szpilrajn extension theorem\cite{davey2002introduction}, every poset has a set of linear extensions. Intuitively, this means we can ``stretch'' a poset into a linear order. This is a powerful result that allows us to study posets by looking at their linearizations.

In this paper, we address the poset cover problem. The poset cover problem is a combinatorial problem that involves finding a minimum set of partial orders whose linear extensions cover a given set of linear orders.

The poset cover problem is known to be NP-complete\cite{heath2013poset}. We hence chose to leverage the power of modern SAT solvers to tackle this problem.

We propose a novel and non-trivial method to reduce the poset cover problem to a Boolean satisfiability problem by utilizing the concept of swap graphs. This approach efficiently handles problem constraints and avoids the exponential growth of naive encodings. Our method can solve instances of the poset cover problem with a reasonable universe size.

Our experimental results demonstrate the effectiveness of the proposed method. We extensively tested using the Z3\cite{de2008z3} theorem prover in Python on randomly generated inputs with varying universe and input sizes. The results show that our method performs well and could be further optimized by a simple divide-and-conquer heuristic.

Future work includes exploring further optimizations to handle larger inputs and investigating potential applications involving partial orders, such as message sequence charts and formal concept analysis.

\section{Preliminaries}
A \emph{partial order} is a binary relation that is reflexive, antisymmetric, and transitive. A \emph{partially ordered set} or \emph{poset} is a binary relational structure $\poset{P} = \tuple{\uni, \leq}$ where the \emph{universe} $\uni$ is a set and $\leq$ is a partial order on $\uni$; we refer to members of $\uni$ as elements of $\poset{P}$ and, where specificity is desired, to $\uni$ as $\uni[
\poset{P}]$ and $\leq$ as $\leq[\poset{P}]$.

\begin{example}
    For the following definition, consider the poset $\poset{P}$ over $\set{a,b,c,d}$ with $a \leq b$\ ; $a \leq c$\ ; $a \leq d$\ ; $b \leq d$\ ; and $c \leq d$.
    \label{example:posetp}
\end{example}

The \emph{cover} relation $\covered$ of a poset $\poset{P}$ is the transitive reduction of its order relation; it describes the case of immediate successor: for $x, y \!\in\! \poset{P}$, $x \covered y$ if and only if $x \leq y$ and there is no $z \!\in\! \poset{P}$ such that $x \lneq z$ and $z \lneq y$. For Example~\ref{example:posetp}, the covered relation includes $a \covered b$\ ; $a \covered c$\ ; $b \covered d$\ ; and $c \covered d$. Note that $\tuple{a,d}$ is absent in $\covered$.

Notice that a poset is equivalent to an acyclic directed graph. Figure~\ref{figure:posetp} describes $\poset{P}$ with a graph $\tuple{V,E} = \tuple{\Omega,\covered}$. This is sometimes called the \emph{Hasse diagram} of a poset.

\begin{figure}[h]
    \centering
    \begin{tikzpicture}
        [
        vertex/.style={circle,thick,draw,minimum size=2em},
        edge/.style={->,> = latex'}
        ]
    \node[vertex] (1) at (1,2) {$a$};
    \node[vertex] (2) at (0,1) {$b$};
    \node[vertex] (3) at (2,1) {$c$};
    \node[vertex] (4) at (1,0) {$d$};
    \draw[edge] (1) -- (2);
    \draw[edge] (1) -- (3);
    \draw[edge] (2) -- (4);
    \draw[edge] (3) -- (4);
    \end{tikzpicture}
    \caption{Graph representation of $\poset{P}$ from Example~\ref{example:posetp}.}
    \label{figure:posetp}
\end{figure}

\begin{example}
    For the following definition, consider the poset $\poset{L}$ over $\set{a,b,c,d}$ with $a \leq b$\ ; $a \leq c$\ ; $a \leq d$\ ; $b \leq c$\ ; $b \leq d$\ ; and $c \leq d$. Figure~\ref{figure:posetl} represents $\poset{L}$ with a graph $\tuple{V,E} = \tuple{\uni,\covered}$.
    \label{example:posetl}
\end{example}

\begin{figure}[h]
    \centering
    \begin{tikzpicture}
        [
        vertex/.style={circle,thick,draw,minimum size=2em},
        edge/.style={->,> = latex'}
        ]
    \node[vertex] (1) at (0,3) {$a$};
    \node[vertex] (2) at (0,2) {$b$};
    \node[vertex] (3) at (0,1) {$c$};
    \node[vertex] (4) at (0,0) {$d$};
    \draw[edge] (1) -- (2);
    \draw[edge] (2) -- (3);
    \draw[edge] (3) -- (4);
    \end{tikzpicture}
    \caption{Graph representation of $\poset{L}$ from Example~\ref{example:posetl}.}
    \label{figure:posetl}
\end{figure}

A partial order where every pair of elements is comparable is called a \emph{linear order}, and a poset $\poset{L}$ with such an order is called a \emph{linear poset}; that is, for $x, y \!\in\! \poset{L}$, $x \leq y$ or $y \leq x$. Note that the graph describing a linear poset is a path. For simplicity, we represent a linear poset in string form. For Example~\ref{example:posetl}, we write \lin{abcd}.

A linear poset $\poset{L}$ that extends a poset $\poset{P}$ is called a \emph{linear extension} or \emph{linearization} of $\poset{P}$, denoted $\poset{P} \lext \poset{L}$; that is, for posets $\poset{P},\poset{L}$ with $\uni[P] \!=\! \uni[L]$, $\poset{P} \lext \poset{L}$ if and only if $\poset{L}$ is linear and $\leq[P] \>\subseteq\> \leq[L]$. For example, for $\poset{P}$ from Example~\ref{example:posetp} and $\poset{L}$ from Example~\ref{example:posetl}, we have $\poset{P} \lext \poset{L}$. Note that a linearization of a poset is equivalent to a topological sort of the graph describing that poset.

Every poset admits a set of linearizations. The set of all linearizations of a poset $\poset{P}$ is denoted $\lang{\poset{P}}$. We shall consider it the \emph{language} of $\poset{P}$. For $\poset{P}$ from Example~\ref{example:posetp}, we have $\lang{\poset{P}} = \set{\lin{abcd},\lin{acbd}}$. For $\poset{L}$ from Example~\ref{example:posetl}, $\lang{\poset{L}} = \set{\lin{abcd}}$.

Now, we are ready to define the \emph{poset cover problem}.

\begin{definition}[Poset Cover Problem]
    Given a set of linear posets $\Upsilon$, find a set of partial orders $C$, called a cover, such that $\abs{C}$ is minimal and the union of the languages of posets in $C$ is equal to $\Upsilon$; that is, $\Upsilon = \bigcup_{\poset{P} \in C} \lang{\poset{P}}$.

    Note that the poset cover problem is equivalent to finding a minimal set of graphs whose topological sorts yield the given set of paths.
    \label{definition:pcp}
\end{definition}

\begin{example}
    Given the set of linear posets $\set{\lin{abdce},\lin{badce},\lin{abcde},\lin{abdec}}$ over $\set{a,b,c,d,e}$, a minimal poset cover of two posets $\poset{A},\poset{B}$ is shown in Figure~\ref{figure:cover example c} with ${\lang{\poset{A}} = \set{\lin{abdce},\lin{abcde},\lin{abdec}}}$ and ${\lang{\poset{B}} = \set{\lin{badce}}}$.
    \label{example:cover example}
\end{example}

\begin{figure}[h]
    \centering
    \begin{subfigure}[b]{0.5\textwidth}
        \centering
        \begin{subfigure}[b]{0.4\textwidth}
            \centering
            \begin{tikzpicture}
                [
                vertex/.style={circle,thick,draw,minimum size=2em},
                edge/.style={->,> = latex'}
                ]
            \node[vertex] (1) at (1,4) {$a$};
            \node[vertex] (2) at (1,3) {$b$};
            \node[vertex] (3) at (0.5,2) {$d$};
            \node[vertex] (4) at (0.5,1) {$e$};
            \node[vertex] (5) at (1.5,2) {$c$};
            \draw[edge] (1) -- (2);
            \draw[edge] (2) -- (3);
            \draw[edge] (3) -- (4);
            \draw[edge] (2) -- (5);
            \end{tikzpicture}
            \caption*{Poset $\poset{A}$}
        \end{subfigure}%
        \begin{subfigure}[b]{0.4\textwidth}
            \centering
            \begin{tikzpicture}
                [
                vertex/.style={circle,thick,draw,minimum size=2em},
                edge/.style={->,> = latex'}
                ]
            \node[vertex] (6) at (4,4) {$b$};
            \node[vertex] (7) at (4,3) {$a$};
            \node[vertex] (8) at (4,2) {$d$};
            \node[vertex] (9) at (4,1) {$c$};
            \node[vertex] (10) at (4,0) {$e$};
            \draw[edge] (6) -- (7);
            \draw[edge] (7) -- (8);
            \draw[edge] (8) -- (9);
            \draw[edge] (9) -- (10);
            \end{tikzpicture}
            \caption*{Poset $\poset{B}$}
        \end{subfigure}
        \caption{A minimal cover.}
        \label{figure:cover example c}
    \end{subfigure}%
    \begin{subfigure}[b]{0.5\textwidth}
        \centering
        \begin{subfigure}[b]{0.4\textwidth}
            \centering
            \begin{tikzpicture}
                [
                vertex/.style={circle,thick,draw,minimum size=2em},
                edge/.style={->,> = latex'}
                ]
            \node[vertex] (1) at (1,4) {$a$};
            \node[vertex] (2) at (1,3) {$b$};
            \node[vertex] (3) at (0.5,2) {$d$};
            \node[vertex] (4) at (0.5,1) {$e$};
            \node[vertex] (5) at (1.5,2) {$c$};
            \draw[edge] (1) -- (2);
            \draw[edge] (2) -- (3);
            \draw[edge] (3) -- (4);
            \draw[edge] (2) -- (5);
            \end{tikzpicture}
            \caption*{Poset $\poset{C}$}
        \end{subfigure}%
        \begin{subfigure}[b]{0.4\textwidth}
            \centering
            \begin{tikzpicture}
                [
                vertex/.style={circle,thick,draw,minimum size=2em},
                edge/.style={->,> = latex'}
                ]
            \node[vertex] (6) at (3.5,4) {$a$};
            \node[vertex] (7) at (4.5,4) {$b$};
            \node[vertex] (8) at (4,3) {$d$};
            \node[vertex] (9) at (4,2) {$c$};
            \node[vertex] (10) at (4,1) {$e$};
            \draw[edge] (6) -- (8);
            \draw[edge] (7) -- (8);
            \draw[edge] (8) -- (9);
            \draw[edge] (9) -- (10);
            \end{tikzpicture}
            \caption*{Poset $\poset{D}$}
        \end{subfigure}
        \caption{Another minimal cover.}
        \label{figure:cover example c'}
    \end{subfigure}
    \caption{Minimal covers for the set $\set{\lin{abdce},\lin{badce},\lin{abcde},\lin{abdec}}$.}
    \label{figure:cover example}
\end{figure}

A minimal poset cover may not be unique, and the languages of the posets in a cover may overlap. For Example~\ref{example:cover example}, there is another minimal poset cover, shown in Figure~\ref{figure:cover example c'}, with overlapping languages ${\lang{\poset{C}} = \set{\lin{abdce},\lin{abcde},\lin{abdec}}}$ and ${\lang{\poset{D}} = \set{\lin{abdce},\lin{badce}}}$.

\section{Reduction to SAT}
Since the poset cover problem is proved to be NP-complete \cite{heath2013poset}, we reduce the problem to SAT to utilize the power of modern SAT solvers. However, a naive encoding would easily result in superpolynomial transformation. We shall first show where the naive encoding falls short and then present a more efficient method.

\subsection{Basic Constraints}
Before getting into constraints given by the problem, here we give the definitions and constraints for the basic poset construction.

\subsubsection{Variable Semantics}
For a poset $\poset{P}$ and $x \!\neq\! y \!\in\! \uni$, we encode with the propositional variable $\satvar{x,y}{\poset{P}}$ the case that $x \leq[\poset{P}] y$.

\subsubsection{Poset Axioms}
Only antisymmetry and transitivity are relevant in our reduction. For a poset $\poset{P}$, we enforce antisymmetry with
\[
\bigwedge_{x \neq y \in \uni} \neg (\satvar{x,y}{\poset{P}} \wedge \satvar{y,x}{\poset{P}})
\]
and transitivity with
\[
\bigwedge_{x \neq y \neq z \in \uni}
\satvar{x,y}{\poset{P}} \wedge \satvar{y,z}{\poset{P}} \implies \satvar{x,z}{\poset{P}}
\]
Moreover, if $\poset{P}$ is linear, we add, on top of that
\[
\bigwedge_{x \neq y \in \uni} \satvar{x,y}{\poset{P}} \vee \satvar{y,x}{\poset{P}}
\]

\subsubsection{Linearization Constraints}
For posets $\poset{P},\poset{L}$ with $\poset{P} \lext \poset{L}$, we encode axioms for both $\poset{P}$ and $\poset{L}$ and add
\[
\bigwedge_{x \neq y \in \uni} \satvar{x,y}{\poset{P}} \implies \satvar{x,y}{\poset{L}}
\]
Within logical formulae, we shall abuse the notation $\poset{P} \lext \poset{L}$ as a macro for the above formula.
\label{subsubsec:linear}

\subsection{Problem Constraints}
Our goal is to find a minimal set $C$ of posets such that $\bigcup_{\poset{P} \in C} \lang{\poset{P}}$ is equal to the input set of linear posets $\Upsilon$. To ensure we find the minimal cover, we encode the case when $\abs{C} = k$, starting from $k = 1$, and incrementally query the SAT solver.

The poset cover problem can now be restated in logical formulae. Given $\Upsilon$ and the size of the cover $k$, we first construct $k$ posets by encoding their axioms. Next, we encode axioms for all the linear posets in $\Upsilon$ and then specify their given order relations accordingly.

We shall encode the problem constraints, $\Upsilon = \bigcup_{\poset{P} \in C} \lang{\poset{P}}$, in a twofold manner, by separately constraining $\Upsilon \subseteq \bigcup_{\poset{P} \in C} \lang{\poset{P}}$ and $\Upsilon \supseteq \bigcup_{\poset{P} \in C} \lang{\poset{P}}$.

\subsubsection{Naive Method} A straightforward way to encode the problem follows immediately from the definition.

For $\Upsilon \subseteq \bigcup_{\poset{P} \in C} \lang{\poset{P}}$, it effectively states that every $\poset{L} \!\in\! \Upsilon$ is a linearization of some $\poset{P} \!\in\! C$. We encode it as
\[
\bigwedge_{\poset{L} \in \Upsilon} \bigvee_{\poset{P} \in C} \poset{P} \lext \poset{L}
\]

For $\Upsilon \supseteq \bigcup_{\poset{P} \in C} \lang{\poset{P}}$, conversely, it states that every $\poset{L} \!\notin\! \Upsilon$ is not a linearization of all $\poset{P} \!\in\! C$. This is encoded as
\[
\bigwedge_{\poset{L} \in \complmt{\Upsilon}} \bigwedge_{\poset{P} \in C} \poset{P} \not\lext \poset{L}
\]

\begin{theorem}
    The constructed formulae are satisfiable if and only if the given case has a solution.
\end{theorem}

Notice that encoding $\Upsilon \supseteq \bigcup_{\poset{P} \in C} \lang{\poset{P}}$ this way would result in exponential blow-up from $\complmt{\Upsilon}$ as there are $\abs{\uni}!$ permutations and hence possible linear posets over $\uni$.

\subsubsection{Swap Graph Method}
Since $\Upsilon \subseteq \bigcup_{\poset{P} \in C} \lang{\poset{P}}$ can be efficiently encoded with naive method, we improve on encoding the other direction. We have devised, by building upon the idea of swap graph, a more efficient reduction that requires some preprocessing using notions defined as follows.

The \emph{adjacent transposition}, or \emph{swap}, relation $\swap$ describes the case of ``off by one swap'' between linear posets with shared universe: for linear posets $\poset{L}_1,\poset{L}_2$, $\poset{L}_1 \swap \poset{L}_2$ if and only if there are $x, y \!\in\! \uni$ such that $\leq[\poset{L}_1] \cap \leq[\poset{L}_2] = {(\leq[\poset{L}_1] \cup \leq[\poset{L}_2])} - \set{\tuple{x,y},\tuple{y,x}}$; to specify which $x,y$ induce the relation, we write $\swap[x,y]$. For example, \lin{abcd} $\swap[b,c]$ \lin{acbd}. Note that $\swap$ is symmetric.

For a set of linear posets $\Upsilon$, the \emph{swap graph} $\sgraph{\Upsilon}$ of it is the undirected graph $\tuple{V,E} = \tuple{\Upsilon,\swap}$. Figure~\ref{figure:graphlp} shows the swap graph $\sgraph{\lang{\poset{P}}}$ for $\poset{P}$ from Example~\ref{example:posetp}. Conveniently, a swap graph built from the language of a poset is connected \cite{ruskey1992generating,pruesse1991generating,heath2013poset}. An iteration of the proof is also in \ref{appendix:connected-proof}.

\begin{figure}[h]
    \centering
    \begin{tikzpicture}
        [
        vertex/.style={circle,thick,draw,minimum size=2em},
        edge/.style={thick}
        ]
    \node[vertex] (1) at (0,0) {\lin{abcd}};
    \node[vertex] (2) at (2,0) {\lin{acbd}};
    \draw[edge] (1) -- (2);
    \end{tikzpicture}
    \caption{Swap graph $\sgraph{\lang{\poset{P}}}$ for $\poset{P}$ from Example~\ref{example:posetp}.}
    \label{figure:graphlp}
\end{figure}

We can exploit this property that a swap graph of poset language is connected to avoid exponential blow-up, by ``insulating'' each strongly connected component $\upsilon \subseteq \sgraph{\Upsilon}$ with a set $moat(\upsilon) = \set{\poset{L} \!\notin\! \upsilon \mid \exists \poset{L}' \!\in\! \upsilon \  \poset{L}' \swap \poset{L}}$ that encloses it. We denote by $moat(\Upsilon)$ the set $\bigcup_{\upsilon \in comp(\Upsilon)} moat(\upsilon)$, where $comp(\Upsilon)$ is the set of all strongly connected components in $\Upsilon$. Figure~\ref{figure:moat} illustrates this idea of moats for two strongly connected components.

\begin{figure}[h]
    \centering
    \begin{tikzpicture}
        \node[draw,ellipse,minimum height=2.5cm,minimum width=1.5cm,fill=white] (comp1) at (8,4) {comp1};

        \node[draw,ellipse,minimum height=2.5cm,minimum width=3.5cm,fill=white] (comp2) at (2,2) {comp2};

        \begin{pgfonlayer}{background}
            \node[draw,ellipse,minimum height=3.5cm,minimum width=2.5cm,fill=gray,fill opacity=0.2,text opacity=1] (comp1moat) at (8,4) {};
            \node[above=0.5mm] at (comp1moat.south) {moat1};

            \node[draw,ellipse,minimum height=3.5cm,minimum width=4.5cm,fill=gray,fill opacity=0.2,text opacity=1] (comp2moat) at (2,2) {};
            \node[above=0.5mm] at (comp2moat.south) {moat2};
        \end{pgfonlayer}

        \node[draw,fit=(comp1moat)(comp2moat)] (x) {};
    \end{tikzpicture}
    \caption{Illustration of moats around strongly connected components.}
    \label{figure:moat}
\end{figure}

An alternative way to encode $\Upsilon \supseteq \bigcup_{\poset{P} \in C} \lang{\poset{P}}$ is then
\[
\bigwedge_{\poset{L} \in moat(\Upsilon)} \bigwedge_{\poset{P} \in C} \poset{P} \not\lext \poset{L}
\]

The construction of $\sgraph{\Upsilon}$ and $moat(\Upsilon)$ together costs $\bigo{\abs{\Upsilon} \!\cdot\! \abs{\uni}}$ by going through each $\poset{L} \!\in\! \Upsilon$ and checking if $\poset{L}' \!\in\! \Upsilon$ for all $\poset{L}' \swap[x,y] \poset{L}$ where $x \covered[\poset{L}] y$, assuming constant membership query of $\Upsilon$.

\begin{proposition}
    Swap graph method is equivalent to naive method.
\end{proposition}
\begin{proof}
    Suppose the moat constraints hold. If the naive constraints do not hold, then there are some $\poset{L} \!\notin\! \Upsilon$ and $\poset{P} \!\in\! C$ such that $\poset{P} \lext \poset{L}$. Per assumption, $\poset{L} \!\notin\! moat(\Upsilon)$, but then $\sgraph{\lang{\poset{P}}}$ is not connected, a contradiction. The converse holds by definition, as $moat(\Upsilon) \subseteq \complmt{\Upsilon}$.
\end{proof}

\subsection{Some More Improvements}
\subsubsection{Reverting to Naivety}
When the input linearizations overwhelm the swap graph universe $\abs{moat(\Upsilon)} > \abs{\complmt{\Upsilon}}$ and the moat method becomes inefficient, we can revert back to the naive method. This is especially useful when the input is extremely dense.

\subsubsection{XOR Encoding}
Recall in section \ref{subsubsec:linear}, when encoding linear posets, we had to enforce total comparability by adding more constraints. We can avoid this requirement since it, combined with antisymmetry, is equivalent to an encoding using XOR($\oplus$).
\[
\bigwedge_{x \neq y \in \uni} \neg (\satvar{x,y}{\poset{P}} \wedge \satvar{y,x}{\poset{P}})
\wedge
\bigwedge_{x \neq y \in \uni} \satvar{x,y}{\poset{P}} \vee \satvar{y,x}{\poset{P}}
\bicond
\bigwedge_{x \neq y \in \uni} \satvar{x,y}{\poset{P}} \oplus \satvar{y,x}{\poset{P}}
\]

\subsubsection{Skipping $\poset{L}$}
To encode $\poset{P} \lext \poset{L}$, we had to encode axioms for both posets. However, since the problem input includes complete order relation $\leq[\poset{L}]$ of all given $\poset{L} \!\in\! \Upsilon$, we can skip the encodings for $\poset{L}$ and set each variable constrained by $\poset{L}$ directly on $\poset{P}$, by taking the contrapositive form of $\lext$. This way, it is possible to lose all the variables for $\poset{L}$ and the associated constraints.
\[
\bigwedge_{x \neq y \in \uni} \satvar{x,y}{\poset{P}} \implies \satvar{x,y}{\poset{L}}
\bicond
\bigwedge_{x \neq y \in \uni} \neg \satvar{x,y}{\poset{L}} \implies \neg \satvar{x,y}{\poset{P}}
\bicond
\bigwedge_{x,y \in \complmt{\leq[\poset{L}]}} \neg \satvar{x,y}{\poset{P}}
\]

\subsubsection{Divide and Conquer}
Finally, with swap graph method, we need not encode the entire problem input in one go. Instead, we can tackle each component $\upsilon \!\in\! comp(\Upsilon)$ individually, which further reduces the number of clauses, with trade-off being the number of SAT queries polynomial in $\abs{\Upsilon}$. This is especially efficient when the input is sparse.

\section{A Special Case}
Although, in general, solving the poset cover problem is hard, in the case that the given set of linearizations constitutes the language of a single poset, the problem can be solved in polynomial time. We can consider this special case as an opening jab at the problem, returning if and once it succeeds.

We know, corollarily from Szpilrajn extension theorem\cite{davey2002introduction}, that a poset is exactly the intersection of all its linearizations; that is, $\leq[\poset{P}] = \bigcap_{\poset{L} \in \lang{\poset{P}}} \leq[\poset{L}]$ for any poset $\poset{P}$. Additionally, with Corollary~\ref{corollary:final}, we can determine if a set of linearizations constitutes the language of a single poset in time $\bigo{\abs{\Upsilon} \!\cdot\! \abs{\uni}}$, and the solution to the poset cover problem is then the poset formed by the intersection of all linearizations.

\section{Experimentation}
To test the limits of our method, we experimented on distinct randomly generated inputs using the Z3\cite{de2008z3} theorem prover from Microsoft Research in Python.

As inputs with a sparse swap graph can be quickly divided and reduced to smaller sub-problems, we restricted our inputs to those with a single connected swap graph in order to maximize the the difficulty of the sub-problem.

We ran 100 trials, with solver timeout set to 15 minutes, on each of the different settings of universe size $1 \leq \abs{\uni} \leq 10$ and input size $1 \leq \abs{\Upsilon} \leq 100$. Figure~\ref{figure:exp} shows the number of timeouts out of the 100 trials in each case.

\begin{figure}[h]
    \centering
    \includegraphics[width=0.9\linewidth]{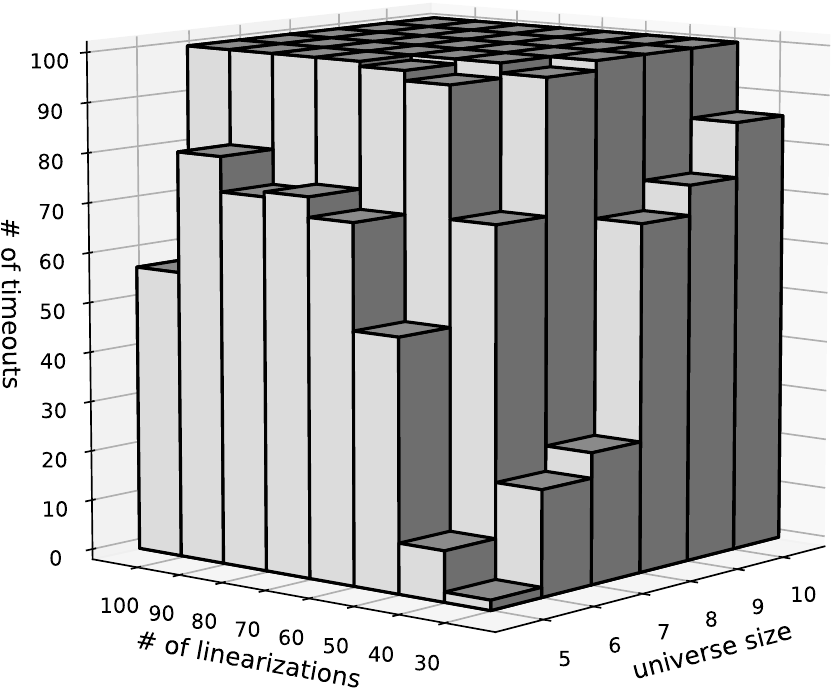}
    \caption{Experimental results.}
    \label{figure:exp}
\end{figure}

Cases where $\abs{\Upsilon} \leq 20$ or $\abs{\uni} \leq 4$ are omitted because they gave no timeouts under extensive testing (over 3000 trials) of combinations with the largest the other parameter: $(\abs{\Upsilon} = 100, \abs{\uni} = 4)$ and $(\abs{\Upsilon} = 20, \abs{\uni} = 10)$.

In the case of the swap graph being a single complex strongly connected component, the number of timeouts increases with the size of the input. These results show that our method is most effective in practice when each component is reasonably small to medium-sized.

In addition, considering the swap graph in practice may be sparse and not a single component, the method can be applied to larger inputs. The divide-and-conquer heuristic can quickly reduce the problem size, and the sub-problems can even be solved in parallel.

\begin{figure}[h]
    \centering
    \includegraphics[width=0.9\linewidth]{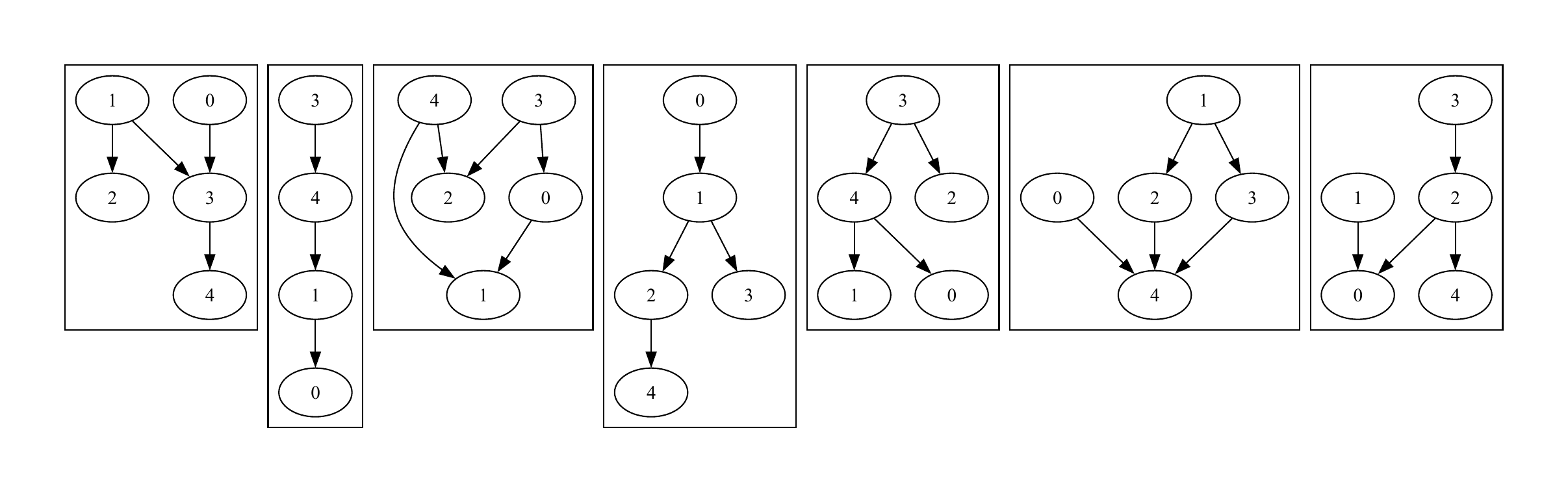}
    \caption{Poset cover of an example with $(\abs{\Upsilon} = 30, \abs{\uni} = 5)$.}
    \label{figure:cover_exp}
\end{figure}

An example poset cover output of a run with $(\abs{\Upsilon} = 30, \abs{\uni} = 5)$ is shown in Figure~\ref{figure:cover_exp}. Its corresponding swap graph is show in Figure~\ref{figure:swap_exp}.

\begin{figure}[h]
    \centering
    \includegraphics[width=0.9\linewidth]{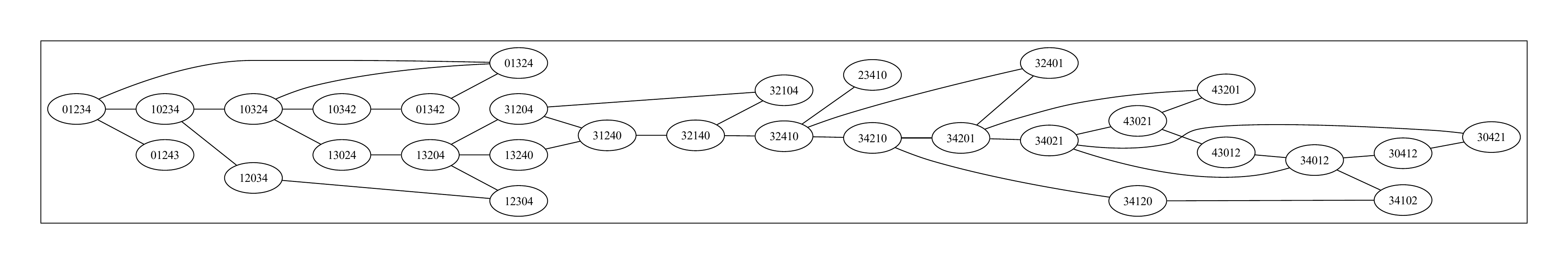}
    \caption{Swap graph of an example with $(\abs{\Upsilon} = 30, \abs{\uni} = 5)$.}
    \label{figure:swap_exp}
\end{figure}

\section{Conclusions}

We presented a novel and non-trivial method to reduce the poset cover problem to the Boolean satisfiability problem. We utilized what we call the ``swap graph'' to avoid exponential blow-up in naive encodings.

Swap graphs enabled a divide-and-conquer heuristic to quickly reduce the problem size. Even under the cases where the inputs cannot be divided into smaller sub-problems, our approach efficiently handles problem constraints and solves instances with a reasonable size.

Experimental results demonstrate the effectiveness of our method in practice. Future work includes exploring further optimizations to handle larger inputs and potential applications where the poset cover problems is relevant, such as message sequence charts and formal concept analysis.

\appendix
\renewcommand{\thesection}{\appendixname}
\section{Details}
\label{appendix:connected-proof}
Here we define the \emph{swapping} function. For a linear poset $\poset{L}$ and some $x,y \!\in\! \uni$ with $x \covered y$, $\swapfn{\poset{L}}{x,y}$ is a linear poset $\poset{L}'$ such that $\poset{L} \swap[x,y] \poset{L}'$.

A \emph{swapping sequence} of length $n$ from a linear poset $\poset{L}$ to a linear poset $\poset{L}'$ is a sequence of linear posets $\poset{L} = \poset{L}_0, \poset{L}_1, ... , \poset{L}_n = \poset{L}'$ obtained by collecting the breadcrumbs of recursively swapping $\poset{L}$ with some sequence of element pairs $x_1,y_1 \ ;\  ... \ ;\  x_n,y_n$ such that $\swapfn{\poset{L}}{x_1,y_1} ... [x_n,y_n] = \poset{L}'$; i.e., $\swapfn{\poset{L}}{x_1,y_1} ... [x_k,y_k] = \poset{L}_k$ for $0 < k < n$.

The length of a minimal swapping sequence from $\poset{L}$ to $\poset{L}'$ is the inversion number $\inv{\poset{L}}{\poset{L}'}$ sorting $\poset{L}$ to $\poset{L}'$, as each swapping in a sequence either increases or decreases the inversion number by $1$ \cite{ruskey1992generating}, and a minimal sequence is where it decreases every step; i.e., $\inv{\poset{L}_{k+1}}{\poset{L}'} = \inv{\poset{L}_k}{\poset{L}'} - 1$ for $0 \leq k < n$. It is also known as the Kendall tau distance, or the bubble-sort distance, between them, which counts the number of discordant pairs, i.e., inversions.

In addition, we describe the case of \emph{incomparability} in a poset $\poset{P}$ with $\incomp$ relation: for $x, y \!\in\! \poset{P}$, $x \incomp y$ if and only if $x \nleq y$ and $y \nleq x$. For Example~\ref{example:posetp}, there is $b \incomp c$.

\begin{lemma}
    For posets $\poset{P},\poset{L}$ with $\poset{P} \lext \poset{L}$, if $\incomp[\poset{P}]\> \neq \emptyset$, then there is $x,y$ with $x \incomp[\poset{P}] y$ such that $x \covered[\poset{L}] y$.
\end{lemma}

\begin{lemma}[\cite{heath2013poset}]
    For posets $\poset{P},\poset{L}$ with $\poset{P} \lext \poset{L}$, if $x \incomp[\poset{P}] y$ and $x \covered[\poset{L}] y$, then there is $\poset{L}'$ such that $\poset{P} \lext \poset{L}'$ and $y \covered[\poset{L}'] x$.
\end{lemma}

\begin{lemma}
    For posets $\poset{P},\poset{L},\poset{L}'$ with $\poset{L} \!\neq\! \poset{L}'$ and $\poset{P} \lext \poset{L},\poset{L}'$, $\leq[\poset{L}] \cap \leq[\poset{L}'] \subseteq\> \incomp[\poset{P}]$ and $\abs{\complmt{\leq[\poset{L}] \cap \leq[\poset{L}']}} = \inv{\poset{L}}{\poset{L}'}$.
\end{lemma}

\begin{lemma}
    For posets $\poset{P},\poset{L},\poset{L}'$ with $\poset{L} \!\neq\! \poset{L}'$ and $\poset{P} \lext \poset{L},\poset{L}'$, their inversion set corresponds to the complement of their intersection.
\end{lemma}

\begin{lemma}[\cite{ruskey1992generating}]
    For posets $\poset{P},\poset{L},\poset{L}'$ with $\poset{L} \!\neq\! \poset{L}'$ and $\poset{P} \lext \poset{L},\poset{L}'$, the distance between $\poset{L}$ and $\poset{L}'$ in $\sgraph{\lang{\poset{P}}}$ is the length of a minimal swapping sequence from $\poset{L}$ to $\poset{L}'$.
\end{lemma}

\begin{lemma}
    For posets $\poset{P},\poset{L},\poset{L}'$ with $\poset{L} \!\neq\! \poset{L}'$ and $\poset{P} \lext \poset{L},\poset{L}'$, for any minimal swapping sequence $\pi$ from $\poset{L}$ to $\poset{L}'$, we have $\poset{P} \lext \poset{L}''$ for all $\poset{L}'' \!\in\! \pi$.
\end{lemma}

\begin{theorem}
    \label{theorem:connected}
    For a poset $\poset{P}$, $\sgraph{\lang{\poset{P}}}$ is connected by exactly all the minimal swapping sequences.
\end{theorem}

\begin{corollary}
    For a set of linear posets $\Upsilon$, there is a poset $\poset{P}$ with $\lang{\poset{P}} = \Upsilon$ if and only if for all $\poset{L}, \poset{L}' \!\in\! \Upsilon$ and for all minimal swapping sequence $\pi$ from $\poset{L}$ to $\poset{L}'$, we have $\poset{L}'' \!\in\! \Upsilon$ for all $\poset{L}'' \!\in\! \pi$.
\end{corollary}

\begin{corollary}
    \label{corollary:final}
    For a set of linear posets $\Upsilon$, define $\poset{P}$ such that $\leq[\poset{P}] = \bigcap_{\poset{L} \in \Upsilon} \leq[\poset{L}]$. $\lang{\poset{P}} = \Upsilon$ if and only if for all $\poset{L} \!\in\! \Upsilon$ where $x \covered[\poset{L}] y$ and $x \incomp[\poset{P}] y$, $\swapfn{\poset{L}}{x,y} \!\in\! \Upsilon$.
\end{corollary}

\begin{credits}
\subsubsection{\ackname}
This project is partially supported by X from Y.
\end{credits}

\bibliography{poset}

\begin{thebibliography}{1}
\providecommand{\url}[1]{\texttt{#1}}
\providecommand{\urlprefix}{URL }
\providecommand{\doi}[1]{https://doi.org/#1}

\bibitem{davey2002introduction}
Davey, B.A., Priestley, H.A.: Introduction to Lattices and Order. Cambridge
  University Press (2002)

\bibitem{de2008z3}
De~Moura, L., Bj{\o}rner, N.: Z3: An efficient smt solver. In: International
  conference on Tools and Algorithms for the Construction and Analysis of
  Systems. pp. 337--340. Springer (2008)

\bibitem{fernandez2013mining}
Fernandez, P.L., Heath, L.S., Ramakrishnan, N., Tan, M., Vergara, J.P.C.:
  Mining posets from linear orders. Discrete Mathematics, Algorithms and
  Applications  \textbf{5}(04),  1350030 (2013)

\bibitem{heath2013poset}
Heath, L.S., Nema, A.K.: The poset cover problem. Open Journal of Discrete
  Mathematics  \textbf{3},  101--111 (2013)

\bibitem{pruesse1991generating}
Pruesse, G., Ruskey, F.: Generating the linear extensions of certain posets by
  transpositions. SIAM Journal on Discrete Mathematics  \textbf{4}(3),
  413--422 (1991)

\bibitem{ruskey1992generating}
Ruskey, F.: Generating linear extensions of posets by transpositions. Journal
  of Combinatorial Theory, Series B  \textbf{54}(1),  77--101 (1992)

\end{thebibliography}
\bibliographystyle{splncs04}

\end{document}